\documentclass[aps,prl,reprint]{revtex4-2}

\usepackage{hyperref}
\hypersetup{%
	colorlinks = true,
    linkcolor={red},
	citecolor={red},
	urlcolor={magenta}
}
\usepackage{xcolor}
\definecolor{apslinkblue}{HTML}{2d3092}

\usepackage[english]{babel} 
\usepackage{amsthm} 
\usepackage{thm-restate}
\usepackage{amssymb} 
\usepackage{tikz-cd} 
\usepackage{graphicx} 
\usepackage{amsmath} 
\usetikzlibrary{angles, quotes}
\usepackage{float}
\usepackage[capitalize]{cleveref} 

\tikzset{%
	vertices/.style={circle,minimum size=12pt,thick,fill=apslinkblue, inner sep = 0pt, text=white},
	edges/.style={draw=darkgray,very thick, shorten <=2pt, shorten >=2pt}
}
\newcommand\scalefigures{0.74} 

\newcommand{\ket}[1]{|{#1} \rangle}
\newcommand{\bra}[1]{\left \langle#1\right |}

\newcommand{\se}{\subseteq}
\newcommand{\ls}{\leqslant}
\newcommand{\gs}{\geqslant}
\newcommand{\sm}{\setminus}

\newcommand{\bigdelta}{\makebox{\large\ensuremath{\Delta}}}

\theoremstyle{plain}
\newtheorem{theorem}{Theorem}
\newtheorem{corollary}{Corollary}
\newtheorem{lemma}{Lemma}

\newtheorem{proposition}{Proposition}

\theoremstyle{remark}

\theoremstyle{definition}
\newtheorem{definition}{Definition}
\newtheorem{example}{Example}

\begin{document}

\title{The 27-qubit Counterexample to the LU-LC Conjecture is Minimal}
\author{Nathan Claudet}
\email{Nathan.Claudet@uibk.ac.at}
\affiliation{University of Innsbruck, Department of Theoretical Physics, Technikerstraße  21a, A-6020 Innsbruck, Austria}

\begin{abstract}
It was once conjectured that two graph states are local unitary (LU) equivalent if and only if they are local Clifford (LC) equivalent. This so-called LU-LC conjecture was disproved in 2007, as a pair of 27-qubit graph states that are LU-equivalent, but not LC-equivalent, was discovered. We prove that this counterexample to the LU-LC conjecture is minimal. In other words, for graph states on up to 26 qubits, the notions of LU-equivalence and LC-equivalence coincide. 
This result is obtained by studying the structure of 2-local complementation, a special case of the recently introduced r-local complementation, and a generalization of the well-known local complementation. We make use of a connection with triorthogonal codes and Reed-Muller codes.
\end{abstract}

\maketitle

Graph states form a versatile family of entangled quantum states, that allow for easy and compact representations thanks to their one-to-one correspondence with mathematical graphs \cite{Hein04, Hein06}. Graph states are universal resources for measurement-based quantum computation \cite{raussendorf2001one, raussendorf2003measurement, briegel2009measurement}, a paradigm introduced by Hans Briegel and Robert Raussendorf in the early 2000s. Graph states first appeared as a generalization of the cluster states \cite{Briegel2001}, the original resources for measurement-based quantum computation \cite{Hein04}. Graph states also arise naturally in the context of quantum error correction \cite{gottesman1997stabilizercodesquantumerror, Grassl2002, schlingemann2001quantum, schlingemann2001stabilizer, khesin2025quantum} and quantum communication networks \cite{markham2008graph, Keet2010, Javelle2013, gravier2013quantum, Bell2014secret, azuma2015all, azuma2023quantum, hahn2019quantum, bravyi2024generating, meignant2019distributing, fischer2021distributing, Mannalath2023}. In these applications, graph states are used as a resource of entanglement. It is thus an essential task to classify graph states according to their entanglement. For general quantum states, having the same entanglement is often formalized at being related by SLOCC (stochastic local operations and classical communication). For graph states in particular, this is the same as being local unitary equivalent, or \textit{LU-equivalent} for short, meaning that the graph states are related by single-qubit unitary operators \cite{Verstraete2003, Hein06}. Thus, classifying graph states according to their entanglement amounts to understanding when two graph states are LU-equivalent.

If we restrict the single-qubit unitaries to be in the so-called Clifford group, this defines a stronger notion of equivalence: the graph states are said local Clifford equivalent, or \textit{LC-equivalent} for short. LC-equivalence of graph states is particularly easy to characterize, as  it is captured by a simple and well-studied graphical operation, called local complementation \cite{VandenNest04}. 
This implies for example the existence of an efficient algorithm for recognizing LC-equivalent graph states \cite{Bouchet1991, VdnEfficientLC}.

Obviously, two LC-equivalent graph states are LU-equivalent.
Conversely, it was once conjectured that two LU-equivalent graph states are always LC-equivalent \cite{5pb}: this is often referred to as the \textit{LU-LC conjecture}. 
More precisely, we write that \textit{LU=LC} holds for a graph state $\ket G$ when every graph state LU-equivalent to $\ket G$, is actually LC-equivalent to $\ket G$: the LU-LC conjecture predicts that LU=LC holds for every graph state. 
Encouraging preliminary results were obtained, providing evidence in support of the LU-LC conjecture \cite{gross2007lu, zeng2011transversality, VandenNest05, Zeng07}. In fact, LU=LC holds for various families of graph states \cite{VandenNest05, Zeng07, Sarvepalli_2010, tzitrin2018local, claudet2024local, hahn2026structurecirclegraphstates}. However, in 2007, Zhengfeng Ji, Jianxin Chen, Zhaohui Wei, and Mingsheng Ying showed that the LU-LC conjecture is false \cite{Ji07}, as they discovered a 27-qubit counterexample, i.e. a pair of 27-qubit graph states that are LU-equivalent but not LC-equivalent. The original 27-qubit counterexample corresponds to a pair of graphs that differ by only one edge. Since, an equivalent counterexample has been found, with a more elegant bipartite form \cite{Tsimakuridze17}. We depict it in Figure \ref{fig:ce27}.

\begin{figure*}[htbp]
\centering
\scalebox{\scalefigures}{
\begin{tikzpicture}
\begin{scope}[every node/.style={circle, fill=blue, inner sep=3pt}]

    \node (u1) at (0,0) {};
    \node (u2) at (1,-0.5) {};
    \node (u3) at (2,-0.75) {};
    \node (u4) at (3,-0.75) {};
    \node (u5) at (4,-0.5) {};
    \node (u6) at (5,0) {};

    \node (v12345) at (-2,5) {};
    \node (v12346) at (-1.25,5) {};
    \node (v12356) at (-0.5,5) {};
    \node (v12456) at (0.25,5) {};
    \node (v13456) at (1,5) {};
    \node (v23456) at (1.75,5) {};

    \node (v1234) at (4,5) {};
    \node (v1235) at (4.75,5) {};
    \node (v2456) at (6.25,5) {};
    \node (v3456) at (7,5) {};    
    
\end{scope}
\begin{scope}[every edge/.style={draw=darkgray,very thick, shorten <=3pt, shorten >=3pt}]
              
    \path [-] (v12345) edge node {} (u1);
    \path [-] (v12345) edge node {} (u2);
    \path [-] (v12345) edge node {} (u3);
    \path [-] (v12345) edge node {} (u4);
    \path [-] (v12345) edge node {} (u5);

    \path [-] (v12346) edge node {} (u1);
    \path [-] (v12346) edge node {} (u2);
    \path [-] (v12346) edge node {} (u3);
    \path [-] (v12346) edge node {} (u4);
    \path [-] (v12346) edge node {} (u6);

    \path [-] (v12356) edge node {} (u1);
    \path [-] (v12356) edge node {} (u2);
    \path [-] (v12356) edge node {} (u3);
    \path [-] (v12356) edge node {} (u5);
    \path [-] (v12356) edge node {} (u6);

    \path [-] (v12456) edge node {} (u1);
    \path [-] (v12456) edge node {} (u2);
    \path [-] (v12456) edge node {} (u4);
    \path [-] (v12456) edge node {} (u5);
    \path [-] (v12456) edge node {} (u6);

    \path [-] (v13456) edge node {} (u1);
    \path [-] (v13456) edge node {} (u3);
    \path [-] (v13456) edge node {} (u4);
    \path [-] (v13456) edge node {} (u5);
    \path [-] (v13456) edge node {} (u6);

    \path [-] (v23456) edge node {} (u2);
    \path [-] (v23456) edge node {} (u3);
    \path [-] (v23456) edge node {} (u4);
    \path [-] (v23456) edge node {} (u5);
    \path [-] (v23456) edge node {} (u6);

    \path [-] (v1234) edge node {} (u1);
    \path [-] (v1234) edge node {} (u2);
    \path [-] (v1234) edge node {} (u3);
    \path [-] (v1234) edge node {} (u4);

    \path [-] (v1235) edge node {} (u1);
    \path [-] (v1235) edge node {} (u2);
    \path [-] (v1235) edge node {} (u3);
    \path [-] (v1235) edge node {} (u5);

    \path [-] (v2456) edge node {} (u2);
    \path [-] (v2456) edge node {} (u4);
    \path [-] (v2456) edge node {} (u5);
    \path [-] (v2456) edge node {} (u6);

    \path [-] (v3456) edge node {} (u3);
    \path [-] (v3456) edge node {} (u4);
    \path [-] (v3456) edge node {} (u5);
    \path [-] (v3456) edge node {} (u6);

\end{scope}
\draw (5.5,5) node(){\large\dots};
\draw [stealth-stealth,draw=black,very thick ](8,2.5) -- (9,2.5);
\begin{scope}[shift={(12,0)}]
\begin{scope}[every node/.style={circle, fill=blue, inner sep=3pt}]

    \node (u1) at (0,0) {};
    \node (u2) at (1,-0.5) {};
    \node (u3) at (2,-0.75) {};
    \node (u4) at (3,-0.75) {};
    \node (u5) at (4,-0.5) {};
    \node (u6) at (5,0) {};

    \node (v12345) at (-2,5) {};
    \node (v12346) at (-1.25,5) {};
    \node (v12356) at (-0.5,5) {};
    \node (v12456) at (0.25,5) {};
    \node (v13456) at (1,5) {};
    \node (v23456) at (1.75,5) {};

    \node (v1234) at (4,5) {};
    \node (v1235) at (4.75,5) {};
    \node (v2456) at (6.25,5) {};
    \node (v3456) at (7,5) {};    
    
\end{scope}
\begin{scope}[every edge/.style={draw=darkgray,very thick, shorten <=3pt, shorten >=3pt}]
              
    \path [-] (v12345) edge node {} (u1);
    \path [-] (v12345) edge node {} (u2);
    \path [-] (v12345) edge node {} (u3);
    \path [-] (v12345) edge node {} (u4);
    \path [-] (v12345) edge node {} (u5);

    \path [-] (v12346) edge node {} (u1);
    \path [-] (v12346) edge node {} (u2);
    \path [-] (v12346) edge node {} (u3);
    \path [-] (v12346) edge node {} (u4);
    \path [-] (v12346) edge node {} (u6);

    \path [-] (v12356) edge node {} (u1);
    \path [-] (v12356) edge node {} (u2);
    \path [-] (v12356) edge node {} (u3);
    \path [-] (v12356) edge node {} (u5);
    \path [-] (v12356) edge node {} (u6);

    \path [-] (v12456) edge node {} (u1);
    \path [-] (v12456) edge node {} (u2);
    \path [-] (v12456) edge node {} (u4);
    \path [-] (v12456) edge node {} (u5);
    \path [-] (v12456) edge node {} (u6);

    \path [-] (v13456) edge node {} (u1);
    \path [-] (v13456) edge node {} (u3);
    \path [-] (v13456) edge node {} (u4);
    \path [-] (v13456) edge node {} (u5);
    \path [-] (v13456) edge node {} (u6);

    \path [-] (v23456) edge node {} (u2);
    \path [-] (v23456) edge node {} (u3);
    \path [-] (v23456) edge node {} (u4);
    \path [-] (v23456) edge node {} (u5);
    \path [-] (v23456) edge node {} (u6);

    \path [-] (v1234) edge node {} (u1);
    \path [-] (v1234) edge node {} (u2);
    \path [-] (v1234) edge node {} (u3);
    \path [-] (v1234) edge node {} (u4);

    \path [-] (v1235) edge node {} (u1);
    \path [-] (v1235) edge node {} (u2);
    \path [-] (v1235) edge node {} (u3);
    \path [-] (v1235) edge node {} (u5);

    \path [-] (v2456) edge node {} (u2);
    \path [-] (v2456) edge node {} (u4);
    \path [-] (v2456) edge node {} (u5);
    \path [-] (v2456) edge node {} (u6);

    \path [-] (v3456) edge node {} (u3);
    \path [-] (v3456) edge node {} (u4);
    \path [-] (v3456) edge node {} (u5);
    \path [-] (v3456) edge node {} (u6);

    \path [-] (u1) edge node {} (u2);
    \path [-] (u1) edge node {} (u3);
    \path [-] (u1) edge node {} (u4);
    \path [-] (u1) edge node {} (u5);
    \path [-] (u1) edge node {} (u6);
    \path [-] (u2) edge node {} (u3);
    \path [-] (u2) edge node {} (u4);
    \path [-] (u2) edge node {} (u5);
    \path [-] (u2) edge node {} (u6);
    \path [-] (u3) edge node {} (u4);
    \path [-] (u3) edge node {} (u5);
    \path [-] (u3) edge node {} (u6);
    \path [-] (u4) edge node {} (u5);
    \path [-] (u4) edge node {} (u6);
    \path [-] (u5) edge node {} (u6);

\end{scope}
\draw (5.5,5) node(){\large\dots};
\end{scope}

\end{tikzpicture}
}
\caption{A 27-qubit counterexample to the LU-LC conjecture \cite{Ji07, Tsimakuridze17}, that is, a pair of graph states that are LU-equivalent but not LC-equivalent. The graphs have 6 bottom vertices. There is one top vertex of degree 5 per set of 5 bottom vertices, and one top vertex of degree 4 per set of 4 bottom vertices; leading to $\binom{6}{5} + \binom{6}{4} = 21$ top vertices. In the leftmost graph, the bottom vertices form an independent set, while in the rightmost graph, the bottom vertices are fully connected. Applying $X(\pi/4)$ on the top qubits and $Z(\pi/4)$ on the bottom qubits maps one graph state to the other. Proving that these two graph states are not LC-equivalent is more involved, a proof can be found in \cite{Tsimakuridze17}.}
\label{fig:ce27}
\end{figure*}
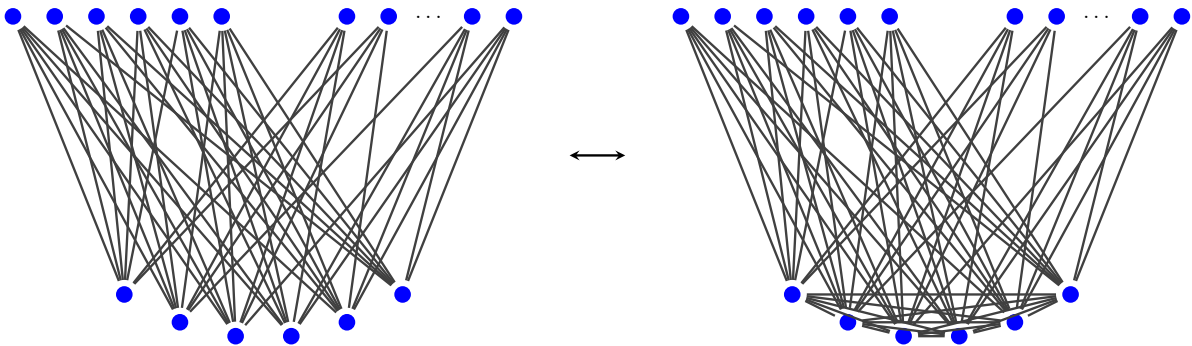

Since 2007, it has been an open question whether this 27-qubit counterexample to the LU-LC conjecture is minimal (in the number of qubits). 
In other words, does LU=LC holds for graph states on up to 26 qubits?

In 2003, by generating every graph on up to 7 vertices, it was proved that LU=LC holds for graph states on up to 7 qubits \cite{Hein04}. 
In 2009, this number was improved to 8 \cite{CABELLO20092219}. Recently, this number was improved to 10 \cite{vandre25} (when graph states are considered up to qubit permutation), then to 11 \cite{burchardt2024algorithm} (also when graph states are considered up to qubit permutation), then to 19 \cite{claudet2025deciding}.

In this paper, we provide a final answer: the 27-qubit counterexample is minimal.

\begin{theorem}\label{thm:lulc26}
    LU=LC holds for graph states on up to 26 qubits.
\end{theorem}

As any stabilizer state is LC-equivalent to a graph state \cite{VandenNest04}, the result extends to stabilizer states.

\begin{corollary}
    LU=LC holds for stabilizer states on up to 26 qubits.
\end{corollary}

As graph states are in one-to-one correspondence with graphs, it theoretically possible \cite{claudet2024local} but in practice out of the question to generate every graph and check whether LU=LC holds for the corresponding graph state. 
Indeed, the number of graphs on up to 26 vertices is of the order of $2^{\binom{26}{2}} \sim 7 \times 10^{97}$, and this number stays very high even when considering unlabeled graphs, i.e. isomorphism classes of graphs: a lower bound is given by $2^{\binom{26}{2}} / 26! \sim 2 \times 10^{71}$.

Our proof instead makes heavy use of the recently introduced formalism of $r$-local complementation. While local complementation captures the LC-equivalence of graph states \cite{VandenNest04}, $r$-local complementation (where $r$ is an integer) captures the LU-equivalence of graph states \cite{claudet2024local}, in the sense that two graph states are LU-equivalent if and only if the corresponding graphs are related by $r$-local complementations. 
In this work, we focus in particular on \textit{2-local complementation}. 
Indeed, 2-local complementation captures the LU-equivalence of graph states on up to 31 qubits.

We begin by giving relevant notations and definitions, before introducing 2-local complementation. Then, we prove \cref{thm:lulc26}. Finally, we discuss open questions.

~\paragraph{Notations.} A \textit{graph} $G = (V,E)$ is composed of two sets, a set $V$ of vertices, and a set $E$ of edges connecting two vertices each, i.e., a subset of $\{\{u,v\} ~|~ u,v \in V\}$. We only consider graphs that are undirected (i.e. edge do not have a direction) and simple (there is at most one edge between two distinct vertices, and no edge connects a vertex to itself). Given $A \se V$, $|A|$ denotes the number of vertices in $A$. 
We use the notation $u \sim_G v$ when $\{u,v\} \in E$, and we say $u$ and $v$ are \textit{adjacent}. Given a vertex $u \in V$, $N_G(u) = \{v\in V ~|~ u \sim_G v\}$ is the \textit{neighborhood} of $u$, i.e., the set of vertices adjacent to $u$. The \textit{degree} of a vertex is the size of its neighborhood. Two vertices $u$ and $v$ that share the same common neighborhood, more precisely $N_G(u)\setminus\{v\}=N_G(v)\setminus\{u\}$, are said to be \textit{twins}. A set of vertices $S$ is said \textit{independent} when no two vertices in $S$ are adjacent. A graph is said \textit{bipartite} if its vertex set can be partitioned into two disjoint independent set.

To any simple undirected graph $G = (V, E)$ on $n$ vertices, is associated an $n$-qubit quantum state $\ket G$, called \textit{graph state}, defined as $\ket G = \left(\prod_{\{u,v\} \in E} CZ_{uv}\right) \ket{+}_V$ where $\ket + = \frac{\ket 0 + \ket 1}{\sqrt 2}$ and $CZ = \ket{00}\bra{00}+\ket{01}\bra{01}+\ket{10}\bra{10}-\ket{11}\bra{11}$.
Alternatively, $\ket G$ is defined as the unique fixpoint (up to global phase) of the operators $X_u Z_{N_G(u)}$ for every $u \in V$, where $Z = \ket{0}\bra{0}-\ket{1}\bra{1}$ and $X = \ket{0}\bra{1}+\ket{1}\bra{0}$.
Single-qubit unitary gates include the \textit{Hadamard gate} $H = \frac{1}{\sqrt 2}(\ket{0}\bra{0}+\ket{0}\bra{1}+\ket{1}\bra{0}-\ket{1}\bra{1})$, the \textit{Z-rotation} 
$Z(\alpha) = \ket 0 \bra 0 + e^{i\alpha} \ket 1 \bra 1$ and the \textit{X-rotation} $X(\alpha) = H Z(\alpha) H$. Single-qubit Clifford gates are those generated by $H$ and $Z(\frac{\pi}{2})$, up to global phase. 

~\paragraph{Local complementation and pivoting.}
Local complementation on a vertex $u$ of a graph consists in complementing the subgraph induced by the neighborhood of $u$. Formally, a local complementation on $u$ maps the graph $G$ to the graph $G\star u= G\Delta K_{N_G(u)}$ where $\Delta$ denotes the symmetric difference on edges and $K_A$ is the complete graph on the vertices of $A \se V$. A local complementation is implemented on a graph state with single-qubit Clifford operators \cite{VandenNest04}: $X(\frac{\pi}{2})_u Z(-\frac{\pi}{2})_{N_G(u)} \ket G = \ket{G \star u}$, conversely if two graph states are LC-equivalent, the corresponding graphs are related by a sequence of local complementations \cite{VandenNest04}. 
Pivoting on an edge $uv$ consists in three successive local complementations: $G\wedge uv = G \star u \star v \star u = G \star v \star u \star v$. For a bipartite graph, pivoting on $uv$ consists in toggling all edges between $N_G(u)\sm\{v\}$ and $N_G(v)\sm\{u\}$, then swapping $u$ and $v$, keeping the graph bipartite (but changing the bipartition). A pivoting is implemented on a bipartite graph state  with two Hadamard gates \cite{van2005edge, mhalla2012graph}: $H_u H_v \ket{G} = \ket{G \wedge uv}$.

~\paragraph{2-local complementation.}

As local complementations on non-adjacent vertices commute, local complementation over an independent set $S$, called 1-local complementation and denoted $G \star^1 S$, is well-defined (and consists in applying a local complementation on every vertex in $S$ in any order). Similarly, 2-local complementation, denoted $G \star^2 S$, is applied over an independent set $S$. Furthermore, there is a condition on $S$ for the 2-local complementation to be valid, called \textit{2-incidence}. Note that in general an $r$-local complementation is defined on a multiset rather than a set, making 2-local complementation easier to introduce than the general case. 
Refer to \cite{claudet2025localequivalencesgraphstates} for an accessible introduction to $r$-local complementation.

\begin{definition}[\cite{claudet2024local}] An independent set $S$ of vertices is called 2-incident if every pair and triplet of vertices has an even number of common neighbors in $S$, i.e. for any distinct $u,v,w \in V \sm S$, $|N_G(u) \cap N_G(v) \cap S| = 0 \bmod 2$ and $|N_G(u) \cap N_G(v) \cap N_G(w) \cap S| = 0 \bmod 2$.
\end{definition}

\begin{definition}[\cite{claudet2024local}]
    A 2-local complementation on a 2-incident independent set of vertices $S$ consists in toggling edges of vertices with a number of common neighbors in $S$ that is $2 \bmod 4$. 
    Formally, $$u\sim_{G\star^2 S} v \Leftrightarrow\left(u\sim_{G} v \oplus |N_G(u)\!\cap\! N_G(v) \!\cap\! S|= 2\bmod 4\right)$$
    where $\oplus$ denotes the logical XOR.
\end{definition}

\begin{example}
    The 27-qubit counterexample to the LU-LC conjecture (see Figure \ref{fig:ce27}) corresponds to a pair of graphs mapped one to another with a 2-local complementation on the 21 top vertices (which form a 2-incident independent set).
\end{example}

\begin{proposition}[\cite{claudet2024local}]
    A 2-local complementation is implemented on a graph state with local unitaries: $$\bigotimes_{u\in S}X\left(\frac {\pi}{4}\right)\bigotimes_{v\in V \sm S}Z\left(-\frac {\pi}{4} |N_G(v) \cap S|\right)\ket{G} = \ket{G\star^2 S}$$
\end{proposition}

As mentioned above, 2-local complementation captures the LU-equivalence of small graph states.

\begin{proposition}[\cite{claudet2025deciding}] \label{prop:lulc_less_than_31_qubits}
    If two graph states on up to 31 qubits are LU-equivalent,
    then the corresponding graphs are related by a sequence of local complementations containing at most one 2-local complementation.
\end{proposition}

In other words, a counterexample to the LU-LC conjecture on up to 31 qubits exhibits a 2-local complementation that cannot be implemented with local complementations. We will prove \cref{thm:lulc26} by showing that this does not occur in graph states on up to 26 qubits.

~\paragraph{Proof of \cref{thm:lulc26}.}
A natural strategy is to generate every independent 2-incident set that can occur in a graph on up to 26 vertices. This is easier than generating every graph on up to 26 vertices, still, there are too many possible independent 2-incident sets for an exhaustive generation in  reasonable time. We thus further reduce the problem.

\begin{restatable}{lemma}{reduction} \label{lemma:reduction}
    Suppose there exists an $n$-qubit counterexample to the LU-LC conjecture where $n \ls 31$. Then, there exists a graph $G$ bipartite with respect to a bipartition $S, V \sm S$ of the vertices such that:
    \begin{enumerate}
        \item $G$ has at most $n+1$ vertices;
        \item the degree of every vertex is odd and at least 3;
        \item no two distinct vertices are twins;
        \item $S$ is 2-incident;
        \item there exists no set $A \se S$ such that $G \star^2 S = G \star^1 A$.
    \end{enumerate}
\end{restatable}

The proof of \cref{lemma:reduction} is constructive, and is provided in the appendix.

\begin{example}
    The bipartite graph corresponding to the 27-qubit counterexample to the LU-LC conjecture (see Figure \ref{fig:ce27}, left), does not satisfy property 2. In fact, only graphs on an even number of vertices may satisfy property 2, according to the handshaking lemma. Applying \cref{lemma:reduction} to this bipartite graph leads to the 28-qubit counterexample to the LU-LC conjecture \cite{Tsimakuridze17}, depicted in Figure \ref{fig:ce28}.
\end{example}

\begin{figure*}[htbp]
\centering
\scalebox{\scalefigures}{
\begin{tikzpicture}
\begin{scope}[every node/.style={circle, fill=blue, inner sep=3pt}]

    \node (u1) at (0,0) {};
    \node (u2) at (1,-0.5) {};
    \node (u3) at (2,-0.75) {};
    \node (u4) at (3,-0.85) {};
    \node (u5) at (4,-0.75) {};
    \node (u6) at (5,-0.5) {};
    \node (u7) at (6,0) {};

    \node (v12345) at (-1,5) {};
    \node (v12346) at (-0.25,5) {};
    \node (v12356) at (0.5,5) {};
    
    \node (v23567) at (5.5,5) {};
    \node (v24567) at (6.25,5) {};
    \node (v34567) at (7,5) {};
    
\end{scope}
\begin{scope}[every edge/.style={draw=darkgray,very thick, shorten <=3pt, shorten >=3pt}]
              
    \path [-] (v12345) edge node {} (u1);
    \path [-] (v12345) edge node {} (u2);
    \path [-] (v12345) edge node {} (u3);
    \path [-] (v12345) edge node {} (u4);
    \path [-] (v12345) edge node {} (u5);

    \path [-] (v12346) edge node {} (u1);
    \path [-] (v12346) edge node {} (u2);
    \path [-] (v12346) edge node {} (u3);
    \path [-] (v12346) edge node {} (u4);
    \path [-] (v12346) edge node {} (u6);

    \path [-] (v12356) edge node {} (u1);
    \path [-] (v12356) edge node {} (u2);
    \path [-] (v12356) edge node {} (u3);
    \path [-] (v12356) edge node {} (u5);
    \path [-] (v12356) edge node {} (u6);

    \path [-] (v23567) edge node {} (u2);
    \path [-] (v23567) edge node {} (u3);
    \path [-] (v23567) edge node {} (u5);
    \path [-] (v23567) edge node {} (u6);
    \path [-] (v23567) edge node {} (u7);
    
    \path [-] (v24567) edge node {} (u2);
    \path [-] (v24567) edge node {} (u4);
    \path [-] (v24567) edge node {} (u5);
    \path [-] (v24567) edge node {} (u6);
    \path [-] (v24567) edge node {} (u7);
    
    \path [-] (v34567) edge node {} (u3);
    \path [-] (v34567) edge node {} (u4);
    \path [-] (v34567) edge node {} (u5);
    \path [-] (v34567) edge node {} (u6);
    \path [-] (v34567) edge node {} (u7);

\end{scope}
\draw (3,5) node(){\large\dots};
\draw [stealth-stealth,draw=black,very thick ](8,2.5) -- (9,2.5);
\begin{scope}[shift={(12,0)}]
\begin{scope}[every node/.style={circle, fill=blue, inner sep=3pt}]

    \node (u1) at (0,0) {};
    \node (u2) at (1,-0.5) {};
    \node (u3) at (2,-0.75) {};
    \node (u4) at (3,-0.85) {};
    \node (u5) at (4,-0.75) {};
    \node (u6) at (5,-0.5) {};
    \node (u7) at (6,0) {};

    \node (v12345) at (-1,5) {};
    \node (v12346) at (-0.25,5) {};
    \node (v12356) at (0.5,5) {};
    
    \node (v23567) at (5.5,5) {};
    \node (v24567) at (6.25,5) {};
    \node (v34567) at (7,5) {};
    
\end{scope}
\begin{scope}[every edge/.style={draw=darkgray,very thick, shorten <=3pt, shorten >=3pt}]
              
    \path [-] (v12345) edge node {} (u1);
    \path [-] (v12345) edge node {} (u2);
    \path [-] (v12345) edge node {} (u3);
    \path [-] (v12345) edge node {} (u4);
    \path [-] (v12345) edge node {} (u5);

    \path [-] (v12346) edge node {} (u1);
    \path [-] (v12346) edge node {} (u2);
    \path [-] (v12346) edge node {} (u3);
    \path [-] (v12346) edge node {} (u4);
    \path [-] (v12346) edge node {} (u6);

    \path [-] (v12356) edge node {} (u1);
    \path [-] (v12356) edge node {} (u2);
    \path [-] (v12356) edge node {} (u3);
    \path [-] (v12356) edge node {} (u5);
    \path [-] (v12356) edge node {} (u6);

    \path [-] (v23567) edge node {} (u2);
    \path [-] (v23567) edge node {} (u3);
    \path [-] (v23567) edge node {} (u5);
    \path [-] (v23567) edge node {} (u6);
    \path [-] (v23567) edge node {} (u7);
    
    \path [-] (v24567) edge node {} (u2);
    \path [-] (v24567) edge node {} (u4);
    \path [-] (v24567) edge node {} (u5);
    \path [-] (v24567) edge node {} (u6);
    \path [-] (v24567) edge node {} (u7);
    
    \path [-] (v34567) edge node {} (u3);
    \path [-] (v34567) edge node {} (u4);
    \path [-] (v34567) edge node {} (u5);
    \path [-] (v34567) edge node {} (u6);
    \path [-] (v34567) edge node {} (u7);

    \path [-] (u1) edge node {} (u2);
    \path [-] (u1) edge node {} (u3);
    \path [-] (u1) edge node {} (u4);
    \path [-] (u1) edge node {} (u5);
    \path [-] (u1) edge node {} (u6);
    \path [-] (u1) edge node {} (u7);
    \path [-] (u2) edge node {} (u3);
    \path [-] (u2) edge node {} (u4);
    \path [-] (u2) edge node {} (u5);
    \path [-] (u2) edge node {} (u6);
    \path [-] (u2) edge node {} (u7);
    \path [-] (u3) edge node {} (u4);
    \path [-] (u3) edge node {} (u5);
    \path [-] (u3) edge node {} (u6);
    \path [-] (u3) edge node {} (u7);
    \path [-] (u4) edge node {} (u5);
    \path [-] (u4) edge node {} (u6);
    \path [-] (u4) edge node {} (u7);
    \path [-] (u5) edge node {} (u6);
    \path [-] (u5) edge node {} (u7);
    \path [-] (u6) edge node {} (u7);

\end{scope}
\draw (3,5) node(){\large\dots};
\end{scope}

\end{tikzpicture}
}
\caption{A 28-qubit counterexample to the LU-LC conjecture \cite{Tsimakuridze17}, that is, a pair of graph states that are LU-equivalent but not LC-equivalent. The graphs have 7 bottom vertices. The top vertices are all of degree 5, and there is one top vertex per set of 5 bottom vertices, leading to $\binom{7}{5} = 21$ top vertices. In the leftmost graph, the bottom vertices form an independent set, while in the rightmost graph, the bottom vertices are fully connected. Applying $X(\pi/4)$ on the top qubits and $Z(\pi/4)$ on the bottom qubits maps one graph state to the other. Proving that these two graph states are not LC-equivalent is more involved, a proof can be found in \cite{Tsimakuridze17}, and a proof in the formalism of $r$-local complementation can be found in \cite{claudet2024local}. The 27-qubit counterexample to the LU-LC conjecture (see Figure \ref{fig:ce27}) can be recovered by removing one bottom vertex.}
\label{fig:ce28}
\end{figure*}
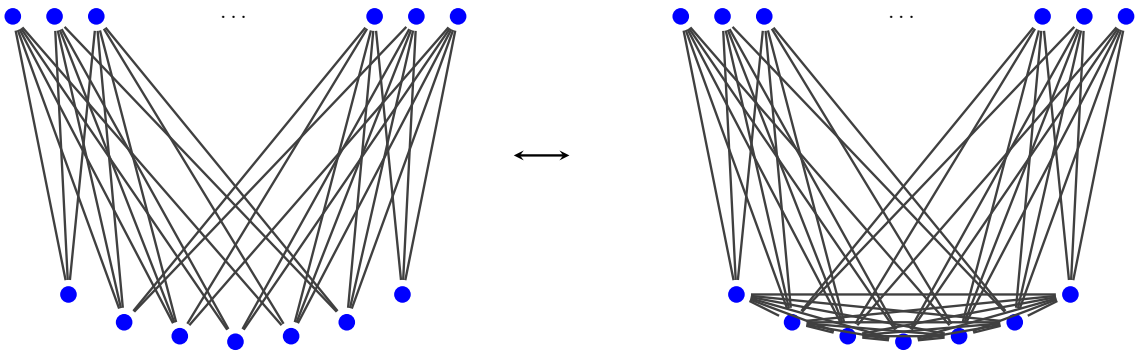

Thanks to \cref{lemma:reduction}, showing that no bipartite graph on up to 27 vertices satisfies properties 2-5, is enough to prove \cref{thm:lulc26}. A natural strategy is to generate every such graph, which may be feasible. It is however not necessary, because of a strong connection between bipartite graphs satisfying properties 2-4 and triorthogonal codes \cite{Bravyi2012, nezami2022, shi2024}, a class of quantum error-correcting codes used in magic state distillation protocols \cite{Bravyi2005, Haah2017magicstate, Haah2018codesprotocols}. This connection is not surprising because, as we will see below, the notions of 2-incidence and triorthogonality are very similar.

To make the connection with triorthogonal codes, to each graph $G$ bipartite with respect to a bipartition $S, V \sm S$ of the vertices, we associate a matrix, whose number of rows is $|V \sm S|$, and whose number of columns is $|S|+ |V \sm S| = |V|$. 
More precisely, $G$ is associated with the matrix $M_G = [~I~|~A~]$ where $I$ is the $|V \sm S| \times |V \sm S|$ identity matrix, and $A$ is the biadjacency matrix of $G$, i.e. $A$ describes the connectivity between $S$ and $V \sm S$. Namely, $A$ has $|V \sm S|$ rows and $|S|$ columns, and for any $a \in S$, $b \in V \sm S$, $A_{ba}=1$ if $a \sim_G b$, and $A_{ba}=0$ otherwise. 
This construction is analogous to the connection between bipartite graphs and binary linear codes \cite{danielsen2008edge}, or binary matroids \cite{OumCliqueWidth}. The latter connection hints at further links between entanglement of graph states and matroid theory.

\begin{restatable}{lemma}{matrix} \label{lemma:matrix}
 The bipartite graph $G$ satisfies properties 2-4 if and only if:
 \begin{itemize}
    \item every column in $M_G$ has odd Hamming weight;
    \item $M_G$ has no repeated columns; 
    \item for any three rows $r_1$, $r_2$ and  $r_3$, $\sum_{i \in [1,|V|]} {r_1}^i {r_2}^i {r_3}^i = 0 \mod 2$, i.e. the product of any three rows has even Hamming weight.
 \end{itemize}
\end{restatable}

Note that the three rows $r_1, r_2, r_3$ need not be distinct. The proof of \cref{lemma:matrix} is provided in the appendix. 
The $\mathbb F_2$-subspace generated by the rows of a matrix satisfying the properties mentioned in \cref{lemma:matrix} is called a \textit{unital triorthogonal subspace} \cite{nezami2022}. More precisely, the triorthogonality corresponds to the property where for any three rows $r_1, r_2, r_3$, $\sum_{i \in [1,k]} {r_1}^i {r_2}^i {r_3}^i = 0 \bmod 2$, and unital means that the subspace contains the all-1 vector, which is true here as the all-1 vector is nothing but the sum of all rows, as each column has odd Hamming weight. Small unital triorthogonal subspaces were classified in \cite{nezami2022}, making use of a strong connection with Reed-Muller codes. There is no canonical choice of a generator matrix for a given unital triorthogonal subspace, however the row operations give all possible generator matrices. Moreover, column permutations generate the isomorphism class of the unital triorthogonal subspace. Note that generator matrices are supposed to have full row-rank, thus the generator matrix of a unital triorthogonal subspace can always be put into the form $M_G$ with row operations and columns permutations.

In the graphical picture, pivoting and vertex permutation correspond to these row operations and columns permutations. More precisely, $M_{G'}$ can be obtained from $M_{G}$ by row operations and columns permutations if and only if $G$ and $G'$ are related by pivotings and vertex permutations \cite{danielsen2008edge}. According to \cref{lemma:matrix}, this proves that properties 2-4 are stable by pivoting and vertex permutation, as well as the following equivalence. 

\begin{lemma}
    There is a one-to-one correspondence between classes of bipartite graphs satisfying properties 2-4 up to pivoting and vertex permutation, and isomorphism classes of unital triorthogonal subspaces.
\end{lemma}

Our strategy to prove \cref{thm:lulc26} is thus to show that if a unital triorthogonal subspace corresponds to graphs on up to 27 vertices, then none of these graphs satisfies property 5. Fortunately, property 5 is stable by pivoting for bipartite graphs satisfying properties 2-4, which we prove in the appendix (see \cref{lemma:pivoting_preserves_properties}), implying that we only need to check a single representative graph for each unital triorthogonal subspace.

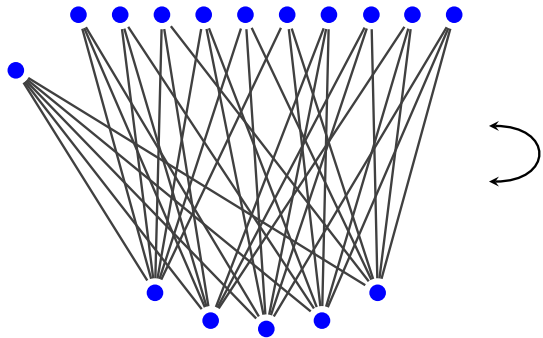
\begin{figure}[htbp]
\centering
\scalebox{\scalefigures}{
\begin{tikzpicture}
\begin{scope}[every node/.style={circle, fill=blue, inner sep=3pt}]

    \newcommand{\drawspacing}{0.75}

    \node (u1) at (0,0) {};
    \node (u2) at (1,-0.5) {};
    \node (u3) at (2,-0.65) {};
    \node (u4) at (3,-0.5) {};
    \node (u5) at (4,0) {};

    \node (v123) at (2-4.5*\drawspacing,5) {};
    \node (v124) at (2-3.5*\drawspacing,5) {};
    \node (v125) at (2-2.5*\drawspacing,5) {};
    \node (v134) at (2-1.5*\drawspacing,5) {};
    \node (v135) at (2-0.5*\drawspacing,5) {};
    
    \node (v145) at (2+0.5*\drawspacing,5) {};
    \node (v234) at (2+1.5*\drawspacing,5) {};
    \node (v235) at (2+2.5*\drawspacing5,5) {};
    \node (v245) at (2+3.5*\drawspacing,5) {};
    \node (v345) at (2+4.5*\drawspacing,5) {};
    
    \node (v12345) at (-2.5,4) {};

\end{scope}
\begin{scope}[every edge/.style={draw=darkgray,very thick, shorten <=3pt, shorten >=3pt}]
    
    \path [-] (v123) edge node {} (u1);
    \path [-] (v123) edge node {} (u2);
    \path [-] (v123) edge node {} (u3);
    
    \path [-] (v124) edge node {} (u1);
    \path [-] (v124) edge node {} (u2);
    \path [-] (v124) edge node {} (u4);
    
    \path [-] (v125) edge node {} (u1);
    \path [-] (v125) edge node {} (u2);
    \path [-] (v125) edge node {} (u5);
    
    \path [-] (v134) edge node {} (u1);
    \path [-] (v134) edge node {} (u3);
    \path [-] (v134) edge node {} (u4);
    
    \path [-] (v135) edge node {} (u1);
    \path [-] (v135) edge node {} (u3);
    \path [-] (v135) edge node {} (u5);

    \path [-] (v145) edge node {} (u1);
    \path [-] (v145) edge node {} (u4);
    \path [-] (v145) edge node {} (u5);
    
    \path [-] (v234) edge node {} (u2);
    \path [-] (v234) edge node {} (u3);
    \path [-] (v234) edge node {} (u4);
    
    \path [-] (v235) edge node {} (u2);
    \path [-] (v235) edge node {} (u3);
    \path [-] (v235) edge node {} (u5);
    
    \path [-] (v245) edge node {} (u2);
    \path [-] (v245) edge node {} (u4);
    \path [-] (v245) edge node {} (u5);
    
    \path [-] (v345) edge node {} (u3);
    \path [-] (v345) edge node {} (u4);
    \path [-] (v345) edge node {} (u5);
    
    \path [-] (v12345) edge node {} (u1);
    \path [-] (v12345) edge node {} (u2);
    \path [-] (v12345) edge node {} (u3);
    \path [-] (v12345) edge node {} (u4);
    \path [-] (v12345) edge node {} (u5);

\end{scope}
\draw[stealth-stealth, draw=black, very thick] (6,2) to[bend right=90, looseness = 3] (6,3);
\end{tikzpicture}
}
\caption{One of the 16-vertex bipartite graphs corresponding to the unique unital triorthogonal subspace in $\mathbb F^{16}_2$, from which the original magic state distillation code by Bravyi and Kitaev \cite{Bravyi2005} and the first code of the Bravyi-Haah family \cite{Bravyi2012} can be derived \cite{nezami2022}. We name the bottom vertices 1, 2, 3, 4 and 5, and we name the top vertices after their neighborhood. One of the top vertices is adjacent to every bottom vertex: $\{1,2,3,4,5\}$. The $\binom{5}{3} = 10$ other top vertices are all possible vertices of degree 3: $\{1,2,3\}$, $\{1,2,4\}$, $\{1,2,5\}$, $\{1,3,4\}$, $\{1,3,5\}$, $\{1,4,5\}$, $\{2,3,4\}$, $\{2,3,5\}$, $\{2,4,5\}$ and $\{3,4,5\}$. The set of top vertices is 2-incident. A 2-local complementation on the top vertices leaves the graph invariant, i.e. does not create or remove any edge. Thus, no counterexample to the LU-LC conjecture can be derived from this graph.}
\label{fig:candidate16}
\end{figure}

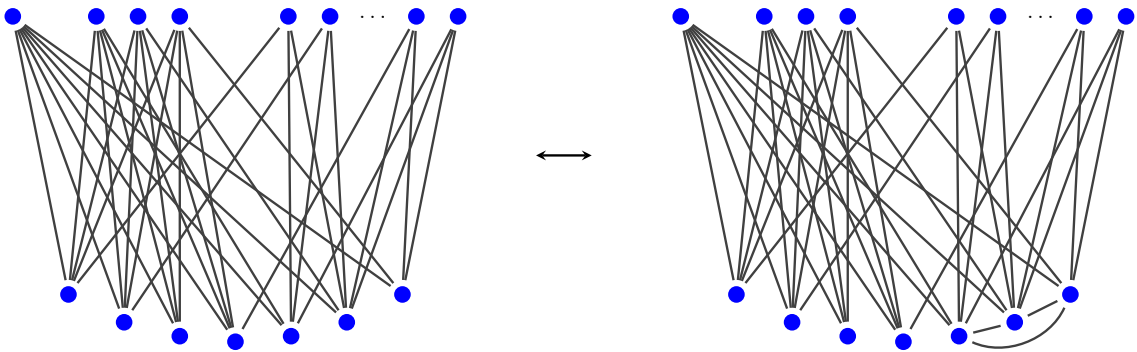
\begin{figure*}[htbp]
\centering
\scalebox{\scalefigures}{
\begin{tikzpicture}
\begin{scope}[every node/.style={circle, fill=blue, inner sep=3pt}]

    \node (u1) at (0,0) {};
    \node (u2) at (1,-0.5) {};
    \node (u3) at (2,-0.75) {};
    \node (u4) at (3,-0.85) {};
    \node (u5) at (4,-0.75) {};
    \node (u6) at (5,-0.5) {};
    \node (u7) at (6,0) {};

    \node (v1234567) at (-1,5) {};
    
    \node (v12345) at (0.5,5) {};
    \node (v12346) at (1.25,5) {};
    \node (v12347) at (2,5) {};
    
    \node (v156) at (3.95,5) {};
    \node (v256) at (4.7,5) {};
    \node (v467) at (6.25,5) {};
    \node (v567) at (7,5) {};
    
\end{scope}
\begin{scope}[every edge/.style={draw=darkgray,very thick, shorten <=3pt, shorten >=3pt}]
    
    \path [-] (v1234567) edge node {} (u1);
    \path [-] (v1234567) edge node {} (u2);
    \path [-] (v1234567) edge node {} (u3);
    \path [-] (v1234567) edge node {} (u4);
    \path [-] (v1234567) edge node {} (u5);
    \path [-] (v1234567) edge node {} (u6);
    \path [-] (v1234567) edge node {} (u7);
    
    \path [-] (v12345) edge node {} (u1);
    \path [-] (v12345) edge node {} (u2);
    \path [-] (v12345) edge node {} (u3);
    \path [-] (v12345) edge node {} (u4);
    \path [-] (v12345) edge node {} (u5);
    
    \path [-] (v12346) edge node {} (u1);
    \path [-] (v12346) edge node {} (u2);
    \path [-] (v12346) edge node {} (u3);
    \path [-] (v12346) edge node {} (u4);
    \path [-] (v12346) edge node {} (u6);
    
    \path [-] (v12347) edge node {} (u1);
    \path [-] (v12347) edge node {} (u2);
    \path [-] (v12347) edge node {} (u3);
    \path [-] (v12347) edge node {} (u4);
    \path [-] (v12347) edge node {} (u7);
    
    \path [-] (v156) edge node {} (u1);
    \path [-] (v156) edge node {} (u5);
    \path [-] (v156) edge node {} (u6);
    
    \path [-] (v256) edge node {} (u2);
    \path [-] (v256) edge node {} (u5);
    \path [-] (v256) edge node {} (u6);
    
    \path [-] (v467) edge node {} (u4);
    \path [-] (v467) edge node {} (u6);
    \path [-] (v467) edge node {} (u7);
    
    \path [-] (v567) edge node {} (u5);
    \path [-] (v567) edge node {} (u6);
    \path [-] (v567) edge node {} (u7);

\end{scope}
\draw (5.5,5) node(){\large\dots};
\begin{scope}[shift={(0.4,0)}]
\draw [stealth-stealth,draw=black,very thick ](8,2.5) -- (9,2.5);
\end{scope}
\begin{scope}[shift={(12,0)}]
\begin{scope}[every node/.style={circle, fill=blue, inner sep=3pt}]

    \node (u1) at (0,0) {};
    \node (u2) at (1,-0.5) {};
    \node (u3) at (2,-0.75) {};
    \node (u4) at (3,-0.85) {};
    \node (u5) at (4,-0.75) {};
    \node (u6) at (5,-0.5) {};
    \node (u7) at (6,0) {};

    \node (v1234567) at (-1,5) {};
    
    \node (v12345) at (0.5,5) {};
    \node (v12346) at (1.25,5) {};
    \node (v12347) at (2,5) {};
    
    \node (v156) at (3.95,5) {};
    \node (v256) at (4.7,5) {};
    \node (v467) at (6.25,5) {};
    \node (v567) at (7,5) {};  
    
\end{scope}
\begin{scope}[every edge/.style={draw=darkgray,very thick, shorten <=3pt, shorten >=3pt}]
              
    \path [-] (v1234567) edge node {} (u1);
    \path [-] (v1234567) edge node {} (u2);
    \path [-] (v1234567) edge node {} (u3);
    \path [-] (v1234567) edge node {} (u4);
    \path [-] (v1234567) edge node {} (u5);
    \path [-] (v1234567) edge node {} (u6);
    \path [-] (v1234567) edge node {} (u7);
    
    \path [-] (v12345) edge node {} (u1);
    \path [-] (v12345) edge node {} (u2);
    \path [-] (v12345) edge node {} (u3);
    \path [-] (v12345) edge node {} (u4);
    \path [-] (v12345) edge node {} (u5);
    
    \path [-] (v12346) edge node {} (u1);
    \path [-] (v12346) edge node {} (u2);
    \path [-] (v12346) edge node {} (u3);
    \path [-] (v12346) edge node {} (u4);
    \path [-] (v12346) edge node {} (u6);
    
    \path [-] (v12347) edge node {} (u1);
    \path [-] (v12347) edge node {} (u2);
    \path [-] (v12347) edge node {} (u3);
    \path [-] (v12347) edge node {} (u4);
    \path [-] (v12347) edge node {} (u7);
    
    \path [-] (v156) edge node {} (u1);
    \path [-] (v156) edge node {} (u5);
    \path [-] (v156) edge node {} (u6);
    
    \path [-] (v256) edge node {} (u2);
    \path [-] (v256) edge node {} (u5);
    \path [-] (v256) edge node {} (u6);
    
    \path [-] (v467) edge node {} (u4);
    \path [-] (v467) edge node {} (u6);
    \path [-] (v467) edge node {} (u7);
    
    \path [-] (v567) edge node {} (u5);
    \path [-] (v567) edge node {} (u6);
    \path [-] (v567) edge node {} (u7);

    \path [-] (u5) edge node {} (u6);
    \path [-] (u6) edge node {} (u7);
    \path [-] (u5) edge[bend right=45] node {} (u7);

\end{scope}
\draw (5.5,5) node(){\large\dots};
\end{scope}

\end{tikzpicture}
}
\caption{The leftmost graph is one of the 24-vertex bipartite graphs corresponding to the unique unital triorthogonal subspace in $\mathbb F^{24}_2$, from which the second code of the Bravyi-Haah family \cite{Bravyi2012} can be derived \cite{nezami2022}. We name the bottom vertices 1, 2, 3, 4, 5, 6 and 7, and we name the top vertices after their neighborhood. One of the top vertices is adjacent to every bottom vertex: $\{1,2,3,4,5,6,7\}$. 3 top vertices are of degree 5: $\{1,2,3,4,5\}$, $\{1,2,3,4,6\}$ and $\{1,2,3,4,7\}$. The $\binom{3}{2}\times\binom{4}{1} + \binom{3}{3} = 3 \times 4 +1 = 11$ other top vertices are of degree 3, they are all possible vertices of degree 3 with 2 or more neighbors among the vertices 4, 5 and 6: $\{1,5,6\}$, $\{2,5,6\}$, $\{3,5,6\}$, $\{4,5,6\}$, $\{1,5,7\}$, $\{2,5,7\}$, $\{3,5,7\}$, $\{4,5,7\}$, $\{1,6,7\}$, $\{2,6,7\}$, $\{3,6,7\}$, $\{4,6,7\}$ and $\{5,6,7\}$. The set of top vertices is 2-incident. A 2-local complementation on the top vertices maps the leftmost graph to the rightmost graph by creating 3 edges: one between 5 and 6, one between 5 and 7, and one between 6 and 7. This transformation can also be obtained with a single local complementation on the vertex $\{5,6,7\}$. Thus, no counterexample to the LU-LC conjecture can be derived from these graphs.
}
\label{fig:candidate24}
\end{figure*}

Surprisingly, up to isomorphism, there are only two unital triorthogonal subspaces that correspond to graphs on up to 27 vertices \cite{nezami2022}. The first one corresponds to a graph on 16 vertices (see Figure \ref{fig:candidate16}), and the second corresponds to a graph on 24 vertices (see Figure \ref{fig:candidate24}). As expected, these two graphs satisfy properties 2-4. However, it is easy to see that no counterexample to the LU-LC conjecture can be derived from these two graphs,  as they do not satisfy property 5. Indeed, in the 16-vertex graph, the 2-local complementation leaves the graph invariant, and in the 24-vertex graph, the 2-local complementation can be implemented with local complementations. This proves \cref{thm:lulc26}, that is, the 27-qubit counterexample to the LU-LC conjecture is minimal.

The 28-qubit counterexample to the LU-LC conjecture (see Figure \ref{fig:ce28}) is recovered from one of the two unital triorthogonal subspaces in $\mathbb F^{28}_2$.

\vspace{5pt}

~\paragraph{Discussion.}

When the LU-LC conjecture was disproved \cite{Ji07}, two main questions remained open: whether the 27-qubit counterexample is minimal, and the existence of an efficient algorithm to decide when two graph states are LU-equivalent. These questions were both addressed with the $r$-local complementation formalism. This present paper gives a final answer to the first question. The second question was partially answered \cite{claudet2025deciding}, but remains open for the time being.

\vspace{14pt}

While the question of the minimal counterexample to the LU-LC conjecture is now closed, the existence of an infinite strict hierarchy of local equivalences between graph states \cite{claudet2024local} paves the way for a generalization of this question. When $r$ is fixed, $r$-local complementation exactly captures when graph states are equivalent up to local unitaries in the level $r+1$ of the so-called Clifford hierarchy \cite{claudet2024local}. It is known that for any $r\gs 2$, there exist graphs related by $r$-local complementations but not $(r-1)$-local complementations. For any integer $r \gs 1$, we may define a number $c(r)$ that is the number of vertices of a minimal pair of graphs related by $r$-local complementations but not $(r-1)$-local complementations, equivalently the number of qubits of a minimal pair of graphs states equivalent up to local unitaries in the level $r+1$ but not $r$ of the Clifford hierarchy. While 0-local complementation is not properly defined, local unitaries in the level 1 of the Clifford hierarchy are Pauli strings, which may only map a given graph state to itself. Thus, we set $c(1)=3$, as $\ket{K_3}$ and $\ket{P_3}$ (where $K_3$ is the triangle, and $P_3$ is the path of length 3) are LC-equivalent but not the same graph state. This present paper proves $c(2)=27$, but no value of $c(r)$ when $r \gs 3$ is currently known. We have however lower and upper bounds, $c(r) \gs 2^{r+2}$ \cite{claudet2025deciding} (when $r \gs 2$) and $c(r) = O(2^{r^2})$ (an exact upper bound is presented in \cite{claudet2024local}). While $c$ grows at least exponentially, the upper bound does not forbid super-exponential growth, thus the question of the growth of $c$ is currently open. Note that $c(r) = \Theta(2^{r^2})$ would imply that LU-equivalence of graph states can be recognized in time $n^{\sqrt{\log_2(n)} + O(1)}$ rather than $n^{\log_2(n) + O(1)}$  (where $n$ is the number of qubits) \cite{claudet2025deciding}. 

\vspace{14pt}

The main result of this paper was obtained through a connection between 2-local complementation and triorthogonal codes. This connection is not surprising, as both are strongly related to the third level of the Clifford hierarchy.  
Many works focus on codes admitting transversal gates at higher levels of the Clifford hierarchy \cite{zeng2011transversality,anderson2014classification,haah2018towers,rengaswamy2020optimality, hu2025climbing,golowich2025quantum,xu2026controlled}: investigating potential connections with $r$-local complementation is a promising direction.

~\paragraph{Acknowledgments.}

We thank Lina Vandré and Piotr Mitosek for discussions on early ideas toward proving this result. We thank Hans Briegel and Robert Raussendorf for interesting discussions on the history of graph states and the LU-LC conjecture. This research was funded by the Austrian Science Fund (FWF) [SFB BeyondC F7102, DOI: 10.55776/F71]. For open access purposes, the authors have applied a CC BY public copyright license to any author accepted manuscript version arising from this submission.

\bibliographystyle{unsrturl}
\bibliography{ref} 

\appendix

\section{Proofs}

We begin by giving some technical lemmata that will be useful for the proof of \cref{lemma:reduction}. \cref{lemma:angle_is_0}, \cref{lemma:1lc_Xrotation} and \cref{lemma:2lc_Xrotation} help simplify proofs involving graph states, informally they show that X-rotations alone define local transformations between graph states. \cref{lemma:pivoting_preserves_properties} shows that properties 4-5 are invariant under pivoting. \cref{lemma:all_odd_after_pivoting} shows how pivoting can be used to obtain bipartite graphs with only vertices of odd degree. \cref{lemma:notwins_nodeg1} presents some properties of bipartite graphs with only vertices of odd degree. 

\begin{lemma}\label{lemma:angle_is_0}
    Given a graph state $\ket{G}$, if $\bigotimes_{v\in V}Z\left(\theta_v\right) \ket{G}$ is also a graph state $\ket{G'}$, then for every $v \in V$, $\theta_v = 0 \bmod 2\pi$ i.e. $G' = G$.
\end{lemma}

\begin{proof}
    For any $a \in V$, $ \bra{0}_{V\sm a}\ket{G'} = \frac{1}{\sqrt{2^{|V|}}}\left(\ket 0_a + \ket 1_a\right)$ (see for example \cite{claudet2025localequivalencesgraphstates}), thus $\bra{0}_{V\sm a}\bra{1}_a\ket{G'} = \frac{1}{\sqrt{2^{|V|}}}$. Also, $$\bra{0}_{V\sm a}\bra{1}_a \bigotimes_{v\in V}Z\left(\theta_v\right)\ket{G} = \frac{ e^{i \theta_a}}{\sqrt{2^{|V|}}}$$
\end{proof}

\begin{lemma}\label{lemma:1lc_Xrotation}
    Given a graph $G$ and an independent set A, 
    if there exist angles $\theta_v$ such that 
    $\bigotimes_{u\in A}X\left(\frac {\pi}{2}\right) \bigotimes_{v\in V \sm A}Z\left(\theta_v\right)\ket{G}$ 
    is a graph state, then the angles $\theta_v$ are uniquely defined, multiples of $\frac{\pi}{2}$ and $\bigotimes_{u\in A}X\left(\frac {\pi}{2}\right) \bigotimes_{v\in V \sm A}Z\left(\theta_v\right)\ket{G} = \ket{G \star^1 A}$.
\end{lemma}

\begin{proof}
    It is known \cite{VandenNest04} that $$\bigotimes_{u\in A}X\left(\frac {\pi}{2}\right) \bigotimes_{v\in V \sm A}Z\left(-\frac {\pi}{2} |N_G(v) \cap A|\right)\ket{G} = \ket{G \star^1 A}$$
    Thus, 
    \begin{align*}
        &\bigotimes_{u\in A}X\left(\frac {\pi}{2}\right) \bigotimes_{v\in V \sm A}Z\left(\theta_v\right)\ket{G}\\
        =& \bigotimes_{v\in V \sm A}Z\left(\theta_v+\frac {\pi}{2} |N_G(v) \cap A|\right)\ket{G \star^1 A}
    \end{align*}
    is a graph state. \cref{lemma:angle_is_0} allows us to conclude. 
\end{proof}

\begin{lemma}\label{lemma:2lc_Xrotation}
    Given a graph $G$ and an independent set S, 
    there exist angles $\theta_v$ such that $\bigotimes_{u\in S}X\left(\frac {\pi}{4}\right) \bigotimes_{v\in V \sm S}Z\left(\theta_v\right)\ket{G}$ is a graph state if and only if $S$ is 2-incident. Also, in that case, the angles $\theta_v$ are uniquely defined, multiples of $\frac{\pi}{4}$ and $\bigotimes_{u\in S}X\left(\frac {\pi}{4}\right) \bigotimes_{v\in V \sm S}Z\left(\theta_v\right)\ket{G} = \ket{G \star^2 S}$.
\end{lemma}

\begin{proof}
    The following are equivalent \cite{claudet2025localequivalencesgraphstates}:
    \begin{itemize}
        \item $S$ is 2-incident;
        \item $\bigotimes_{u\in S}X\left(\frac {\pi}{4}\right)\bigotimes_{v\in V \sm S}Z\left(-\frac {\pi}{4} |N_G(v) \cap S|\right)\ket{G}$ is a graph state;
        \item $\bigotimes_{u\in S}X\left(\frac {\pi}{4}\right)\ket{G}$ is a graph state up to Z-rotations on $V \sm S$.
    \end{itemize}
    Thus, if $\bigotimes_{u\in S}X\left(\frac {\pi}{4}\right)\bigotimes_{v\in V \sm S}Z\left(\theta_v\right)\ket{G}$ is a graph state, then $S$ is 2-incident and 
    \begin{align*}
        &\bigotimes_{u\in S}X\left(\frac {\pi}{4}\right)\bigotimes_{v\in V \sm S}Z\left(\theta_v\right)\ket{G}\\
        =& \bigotimes_{v\in V \sm S}Z\left(\theta_v+\frac {\pi}{4} |N_G(v) \cap S|\right)\ket{G \star^2 S}
    \end{align*}
    is a graph state. \cref{lemma:angle_is_0} allows us to conclude.
\end{proof}

\begin{lemma}\label{lemma:pivoting_preserves_properties}
    Let $G$ be a graph bipartite with respect to a bipartition $S, V \sm S$ of the vertices where $S$ satisfies properties 4-5. Let $a \in S$ and $b \in V\sm S$ such that $a~\sim_G~b$.
    \begin{itemize}
        \item If $b$ has odd degree: let $G' = G \wedge ab$ and $S ' = S \Delta \{a,b\}$;
        \item If $b$ has even degree: let $G' = G \wedge ab - b$ and $S ' = S \sm \{a\}$.
    \end{itemize}
    Then, $G'$ is bipartite with respect to the bipartition $S', {V' \sm S'}$ of the vertices and $S'$ satisfies properties 4-5. 
\end{lemma}

\begin{proof}
    In both cases (b of degree odd or even) we prove, using \cref{lemma:2lc_Xrotation}, that $S'$ is 2-incident in $G \wedge ab$ and that $G \star^2 S \wedge ab = G \wedge ab \star^2 S' \star^1 B$, where $B=\emptyset$ or $\{b\}$.
    \begin{align*}
    &\ket{G\star^2 S \wedge{ab}}\\
    =&H_a H_b \bigotimes_{u\in S}X\left(\frac {\pi}{4}\right)\bigotimes_{v\in V \sm S}Z\left(-\frac {\pi}{4} |N_G(v)|\right)\ket{G}\\
    =& \bigotimes_{u\in B}X\left(\frac {\pi}{2}\right)\bigotimes_{u\in S'}X\left(\frac {\pi}{4}\right)\bigotimes_{v\in V \sm (S' \cup B)}Z\left(\theta_v\right)H_a H_b \ket{G}\\
    =& \bigotimes_{u\in B}X\left(\frac {\pi}{2}\right)\bigotimes_{u\in S'}X\left(\frac {\pi}{4}\right)\bigotimes_{v\in V \sm (S' \cup B)}Z\left(\theta_v\right)\ket{G \wedge ab}\\
    =& \bigotimes_{u\in S'}X\left(\frac {\pi}{4}\right)\bigotimes_{v\in V \sm (S' \cup B)}Z\left(\widetilde{\theta_v}\right)\ket{G \wedge ab \star^1 B}\\
    =& \ket{G \wedge ab \star^1 B \star^2 S'} = \ket{G \wedge ab\star^2 S' \star^1 B}
    \end{align*} 
    
    The last equality holds because $S' \cup B$ is an independent set in $G \wedge ab$ \cite{claudet2024local}. 
    
    In both cases (b of degree odd or even), that translates to $S'$ being 2-incident in $G'$. Then, suppose by contradiction that there exists a set $A' \se S'$ such that $G' \star^2 S' = G' \star^1 A'$. In both cases (b of degree odd or even), this translates to $G \wedge ab \star^2 S' = G \wedge ab \star^1 A'$. We prove, using \cref{lemma:1lc_Xrotation}, that this implies the existence of some set $A \se S$ such that $G \star^2 S = G \star^1 A$.
    \begin{align*}
    &\ket{G \star^2 S}\\
    =&\ket{G \star^2 S \wedge ab \wedge ab}\\
    =&\ket{G \wedge ab \star^2 S' \star^1 B \wedge ab}\\
    =&\ket{G \wedge ab \star^1 A' \star^1 B \wedge ab}\\
    =&H_a H_b \bigotimes_{u\in A' \Delta B}X\left(\frac {\pi}{2}\right)\\
    & \bigotimes_{v\in V \sm (A' \cup B)}Z\left(-\frac {\pi}{2}|N_{G\wedge ab}(v) \cap (A' \Delta B)|\right) H_a H_b \ket{G}\\
    =&\bigotimes_{u\in A}X\left(\frac {\pi}{2}\right)\bigotimes_{v\in V \sm A}Z\left(\theta_v\right)\ket{G}\\
    =&\ket{G \star^1 A}
    \end{align*}
    
    Overall, $A \Delta A'= \emptyset$, $\{a\}$, $\{b\}$ or $\{a,b\}$.
\end{proof}

\begin{lemma}\label{lemma:all_odd_after_pivoting}
    Let $G$ be a graph bipartite with respect to a bipartition $S, V \sm S$ of the vertices. 
    Let $a \in S$ and $b \in V\sm S$ such that $a~\sim_G~b$ and $b$ is the only vertex of even degree in $G$.
    Then, $G \wedge ab - b$ has only vertices of odd degree.
\end{lemma}

\begin{proof}
    The neighborhood of $b$ in $G \wedge ab$ is the same as the neighborhood of $a$ in $G$, more precisely, $N_{G \wedge ab}(b) = N_G(a) \Delta \{a,b\}$. 
    
    First, let us prove that the degree of every neighbor of $b$ in $G \wedge ab$ is even. $|N_{G \wedge ab}(a)| = |N_G(b) \Delta \{a,b\}| = |N_G(b)|+|\{a,b\}| \bmod 2 = 0 \bmod 2$. Also, for any $u \in N_{G \wedge ab}(b) \sm \{a\} = N_G(a) \sm \{b\}$, $|N_{G \wedge ab}(u)| = |N_G(u) \Delta N_G(b) \Delta \{a\}| =|N_G(u)|+| N_G(b)|+|\{a\}| \bmod 2 = 0 \bmod 2$.
    
    Second, the degree of every vertex other than $b$ and its neighborhood in $G \wedge ab$ is odd, in particular, for any $v \in N_{G \wedge ab}(a) \sm \{b\} = N_G(b) \sm \{a\}$, $|N_{G \wedge ab}(v)| = |N_G(v) \Delta N_G(a) \Delta \{b\}| =|N_G(v)|+| N_G(a)|+|\{b\}| \bmod 2 = 1 \bmod 2$.
    
    Thus, after removing $b$ from $G \wedge ab$, each vertex has odd degree.
\end{proof}

\begin{lemma}\label{lemma:notwins_nodeg1}
    Let $G$ be a graph bipartite with respect to a bipartition $S, V \sm S$ of the vertices where each vertex in $G$ has odd degree, also $S$ is 2-incident and contains no vertices of degree 1. Then, $V \sm S$ contains no vertices of degree 1 and no twins.
\end{lemma}

\begin{proof}
    First, suppose by contradiction that $b,c \in {V \sm S}$ are twins. Then, $|N_G(b) \cap N_G(c) \cap S| = |N_G(b)| = 1 \bmod 2$, contradicting the 2-incidence of $S$. 

    Second, suppose by contradiction that $b \in V \sm S$ is adjacent to a unique vertex $a \in S$. As $a$ does not have degree 1, there exists a vertex $c \in V \sm S$ adjacent to $a$. Then, $|N_G(b) \cap N_G(c) \cap S| = |\{a\}|= 1$, contradicting the 2-incidence of $S$.
\end{proof}

We are now ready to prove \cref{lemma:reduction}, which we restate below for convenience.

\reduction*

\begin{proof}

If there exists an $n$-qubit counterexample to the LU-LC conjecture where $n \ls 31$, then there exists a graph $G=(V,E)$ on $n$ vertices and an independent set $S \se V$ satisfying properties 4-5 \cite{claudet2025deciding}. We introduce an algorithm that transforms $G$ into a graph satisfying properties 1-5.

\begin{itemize}
    \item \textbf{Step 1:~}Remove the edges between vertices of $V \sm S$. This makes $G$ bipartite with respect to the bipartition $S, V \sm S$ of the vertices.
    \item \textbf{Step 2:~}Remove each vertex in $S$ of degree 1. Then, until $S$ contains no twins, remove pairs of twins in $S$ (one pair at a time). Finally, remove each vertex in $G$ of degree 0.
    \item \textbf{Step 3:~}If $V \sm S$ contains a vertex $b$ of even degree, choose $a \in S$ such that $a \sim_G b$, then replace $G$ by $G \wedge ab -b$ and $S$ by $S \sm \{a\}$, then go to step 2. 
    \item \textbf{Step 4:~}If $S$ contains vertices of even degree, i.e. if the set $S_{\text{even}}=\{u \in S ~|~ |N_G(u)| = 0 \bmod 2\}$ is not empty, add a vertex $b_{\text{even}}$ to $G$ and connect it to every vertex in $S_{\text{even}}$, then go to step 3.
\end{itemize}

~\paragraph{Correctness. } Step 1 preserves properties 4-5, and after step 1, $G$ remains bipartite. Step 2 preserves properties 4-5. When step 2 is completed, there is no vertex of degree 0 in $G$ and $S$ contains no vertices of degree 1 and no twins. According to \cref{lemma:pivoting_preserves_properties}, step 3 preserves properties 4-5. 
When step 3 is completed,  every vertex in $V \sm S$ has odd degree.
Step 4 preserves that there is no vertex of degree 0 in $G$ and that $S$ contains no vertices of degree 1 and no twins.
Also, we show that step 4 preserves property 4, i.e. $S$ remains 2-incident. We need to check that for any (not necessarily distinct) $u,v \in (V \sm S)\sm b_{\text{even}}$, 
$|N_G(u) \cap N_G(v) \cap S_{\text{even}}| = 0 \bmod 2$. First note that $S_{\text{even}}= S \Delta \left( \bigdelta_{w \in (V \sm S)\sm b_{\text{even}}} N_G(w)\right)$. Then, 
    \begin{align*}
        &\left|N_G(u) \cap N_G(v) \cap S_{\text{even}}\right|\\
        =& |N_G(u) \cap N_G(v)|\\
        &+\sum_{w \in (V \sm S)\sm b_{\text{even}}}|N_G(u) \cap N_G(v) \cap N_G(w)| \bmod 2\\
        =&0 \bmod 2
    \end{align*}
    
    Thus, $S$ is still 2-incident. It is also easy to check that step 4 preserves property 5. When step 4 is completed, the degree of every vertex is odd, and $S$ contains only vertices of degree at least 3 and contains no twins. ${V\sm S}$ contains only vertices of degree at least 3 and contains no twins either, according to \cref{lemma:notwins_nodeg1}. In other words, properties 1-5 are satisfied. 

~\paragraph{Termination. } The number of vertices strictly decreases when the condition of step 3 is met, which guarantees to reach step 4. If the condition of step 4 is not met, the algorithm ends. If the condition of step 4 is met, after adding $b_{\text{even}}$ to $G$, every vertex other than $b_{\text{even}}$ has odd degree. 
If $b_{\text{even}}$ has odd degree, afterward the algorithm ends, as the conditions of both step 3 and 4 are not met. If $b_{\text{even}}$ has even degree, afterward, step 3 chooses $a_{\text{even}} \in S_{\text{even}}$ such that $a_{\text{even}} \sim_G b_{\text{even}}$, then replaces $G$ by $G \wedge a_{\text{even}}b_{\text{even}} -b_{\text{even}}$ and $S$ by $S \sm \{a_{\text{even}}\}$. Every vertex in $G$ is now odd according to \cref{lemma:all_odd_after_pivoting}. Afterward, either step 2 strictly decreases the number of vertices, or the algorithm ends, as the conditions of both step 3 and 4 are not met.
\end{proof}

We also give the proof of \cref{lemma:matrix}, which we restate below for convenience.

\matrix*

\begin{proof}
    It is obvious that if $G$ satisfies properties 2-4, then $M_G$ satisfies these properties above. Conversely, if $M_G$ satisfies these properties above, then $G$ (bipartite according to a bipartition $S, V \sm S$ of the vertices) contains only vertices of odd degree. Also, $S$ contains only vertices of degree at least 3, contains no twins, and is 2-incident. \cref{lemma:notwins_nodeg1} allows us to conclude.
\end{proof}

\end{document}